\documentclass[11pt,reqno]{amsart}

\usepackage[margin=1in]{geometry}
\usepackage[T1]{fontenc}
\usepackage{lmodern}
\usepackage{microtype}
\usepackage{mathtools}
\usepackage{booktabs}
\usepackage{float}
\usepackage{amssymb}
\usepackage{aliascnt}
\usepackage{xcolor}
\usepackage[
  colorlinks=true,
  linkcolor=blue!55!black,
  citecolor=blue!55!black,
  urlcolor=blue!55!black
]{hyperref}
\usepackage[nameinlink,capitalize,noabbrev]{cleveref}

\definecolor{revisioncyan}{RGB}{0,0,0}
\newcommand{\rev}[1]{{\color{revisioncyan}#1}}
\newcommand{\Paa}{{\color{revisioncyan}P_{\mathrm{aa}}}}

\allowdisplaybreaks
\newtheorem{theorem}{Theorem}[section]

\newaliascnt{proposition}{theorem}
\newtheorem{proposition}[proposition]{Proposition}
\aliascntresetthe{proposition}

\newaliascnt{lemma}{theorem}
\newtheorem{lemma}[lemma]{Lemma}
\aliascntresetthe{lemma}

\newaliascnt{corollary}{theorem}
\newtheorem{corollary}[corollary]{Corollary}
\aliascntresetthe{corollary}

\newaliascnt{definition}{theorem}
\newtheorem{definition}[definition]{Definition}
\aliascntresetthe{definition}

\theoremstyle{remark}
\newaliascnt{remark}{theorem}
\newtheorem{remark}[remark]{Remark}
\aliascntresetthe{remark}

\newaliascnt{example}{theorem}

\aliascntresetthe{example}

\crefname{theorem}{Theorem}{Theorems}
\crefname{proposition}{Proposition}{Propositions}
\crefname{lemma}{Lemma}{Lemmas}
\crefname{corollary}{Corollary}{Corollaries}
\crefname{definition}{Definition}{Definitions}
\crefname{remark}{Remark}{Remarks}
\crefname{example}{Example}{Examples}

\DeclareMathOperator{\Tr}{Tr}
\DeclareMathOperator{\rank}{rank}
\DeclareMathOperator{\Alt}{Alt}
\DeclareMathOperator{\SR}{SR}
\DeclareMathOperator{\ran}{ran}

\newcommand{\cL}{\mathcal L}
\newcommand{\cR}{\mathcal R}
\newcommand{\id}{\mathrm{id}}
\newcommand{\HS}{\mathrm{HS}}
\newcommand{\ketbra}[2]{\lvert #1\rangle\!\langle #2\rvert}
\newcommand{\ip}[2]{\langle #1,#2\rangle}
\newcommand{\norm}[1]{\lVert #1\rVert}
\newcommand{\abs}[1]{\lvert #1\rvert}
\newcommand{\C}{\mathbb C}

\hypersetup{
  pdftitle={Two-copy nondistillability of Werner states: sharp partial-trace inequalities and finite-copy extensions},
  pdfauthor={Kishor Bharti, Rishikesh Gajjala, and Tobias Haug},
  pdfsubject={Exact two-copy distillability, arbitrary-copy reformulations, structured witnesses, and finite-copy Werner bounds},
  pdfkeywords={partial trace, Werner state, entanglement distillability, tensor-stable 2-positivity, singular values, reduction map}
}

\title[Two-copy Werner nondistillability and finite-copy bounds]
{Two-copy nondistillability of Werner states:\\
sharp partial-trace inequalities and finite-copy extensions}

\author{Kishor Bharti}
\address{Joint Center for Quantum Information and Computer Science (QuICS),
University of Maryland, College Park, Maryland 20742, USA}
\email{kishor.bharti1@gmail.com}

\author{Rishikesh Gajjala}
\address{Center for Quantum and Topological Systems,
New York University Abu Dhabi, Abu Dhabi, United Arab Emirates}
\email{r.gajjala@nyu.edu}

\author{Tobias Haug}
\address{Quantum Research Center, Technology Innovation Institute,
Abu Dhabi, United Arab Emirates}
\email{tobias.haug@tii.ae}

\date{26 July 2026}

\begin{document}

\begin{abstract}
\begingroup
\color{revisioncyan}
We solve the two-copy distillability problem for Werner states in every local dimension. Our main matrix result is a sharp, dimension-free inequality: for every rank-at-most-two operator, the sum of the squared Hilbert--Schmidt norms of its two partial traces is bounded by twice its squared Hilbert--Schmidt norm plus one half of the squared modulus of its trace. This implies that a Werner state $\rho_\alpha$ is two-copy distillable if and only if $\alpha<-1/2$. In particular, the two-ququart state $\rho^{(4)}_{-1/2}$ is two-copy undistillable, resolving Problem 5 of Horodecki, Rudnicki, and \.Zyczkowski. For an arbitrary finite number $k$ of copies, we give three exact formulations of the remaining problem. At the endpoint $\alpha=-1/2$, undistillability is equivalent to nonnegativity of the endpoint partial-trace form on every rank-at-most-two operator. We also derive an equivalent hierarchy of operator inequalities $H_k(\psi)\succeq0$ for pure states with a maximally mixed qubit marginal. The two-copy proof does not formally induct, because partial trace can increase rank and $2$-positivity is not generally preserved by tensor products. We prove two rigorous many-copy extensions. First, the quadratic form factorizes exactly on tensor-factorized witnesses; for any such decomposition, the endpoint inequality holds if its possible rank-two factor is supported on a block containing at most two copies. Second, we construct explicit constants $\gamma_k>0$ such that $\alpha\ge-\gamma_k$ implies $k$-copy undistillability in every dimension. The initial proofs were generated by ChatGPT 5.6 Sol; the authors have verified and rewritten them to enhance readability and provide additional context.
\endgroup
\end{abstract}

\maketitle
\enlargethispage{3pt}

\section{Introduction}
\label{sec:introduction}

Entanglement distillation asks when several noisy bipartite states can be
converted, by local operations and classical communication, into a smaller
number of high-quality entangled pairs.  A standard criterion says that a
bipartite state $\rho$ is $k$-copy distillable precisely when
$(\rho^\Gamma)^{\otimes k}$ has a negative expectation value on a vector of
Schmidt rank at most two; here $\Gamma$ denotes partial transpose
\cite{HorodeckiEtAl1997,HorodeckiEtAl1998,DurEtAl2000}.  Every state with
positive partial transpose is therefore undistillable.  Whether every state
with negative partial transpose is distillable remains a central open problem
\cite{DiVincenzoEtAl2000,HorodeckiRudnickiZyczkowski2022,
MullerHermesReebWolf2016}.

The Werner family is the canonical symmetric test case.  For $d\ge2$, on
$\C^d\otimes\C^d$ let
\begin{equation}
\label{eq:werner-intro}
 \rho_\alpha=\rho^{(d)}_\alpha
 :=\frac{I_{d^2}+\alpha F_d}{d^2+\alpha d},
 \qquad -1\le \alpha\le 1,
\end{equation}
where $F_d(x\otimes y)=y\otimes x$.  These states are separable exactly in the
positive-partial-transpose range $\alpha\ge-1/d$
\cite{Werner1989,VollbrechtWerner2001}, and they are one-copy distillable
exactly when $\alpha<-1/2$ \cite{DurEtAl2000}.  Thus, for $d\ge3$, the interval
\[
 -\frac12\le \alpha<-\frac1d
\]
consists of NPT states that are not one-copy distillable.  The first genuinely
collective case is two copies.  It has generated semidefinite,
operator-space, matrix-inequality, and partial-trace formulations
\cite{PankowskiEtAl2010,JohnstonKribs2010,CostaRico2025,LiuChen2025,
RicoWolf2025,LiuChen2026}.
The boundary instance $d=4$ and $\alpha=-1/2$ was singled out as
Problem~5 in the list of Horodecki, Rudnicki, and {\.{Z}}yczkowski
\cite{HorodeckiRudnickiZyczkowski2022}.  It is NPT because
$-1/2<-1/4$, but it is not one-copy distillable.  Earlier partial-trace
methods proved two-copy undistillability in a substantial parameter range,
including $\alpha\ge-1/3$ \cite{RicoWolf2025}, but did not reach this
endpoint.

The two-copy algebraic core is the following sharp rank-two inequality.
Throughout, $\norm{\cdot}_2$ denotes the Hilbert--Schmidt norm and
$s_1(X)\ge s_2(X)\ge\cdots$ are the singular values of $X$.

\begin{theorem}[Sharp rank-two partial-trace inequality]
\label{thm:intro-main}
Let $U$ and $V$ be finite-dimensional complex Hilbert spaces.  If
$C\in\cL(U\otimes V)$ and $\rank C\le2$, then
\begin{equation}
\label{eq:intro-main}
 \norm{\Tr_U C}_2^2+\norm{\Tr_V C}_2^2
 \le 2\norm{C}_2^2+\frac12\abs{\Tr C}^2.
\end{equation}
If $\dim U,\dim V\ge2$, the two constants are jointly sharp.
\end{theorem}

It gives the complete two-copy answer.

\begin{corollary}[Exact two-copy threshold]
\label{cor:intro-werner}
For every $d\ge2$ and $-1\le\alpha\le1$, the Werner state $\rho_\alpha$ is
two-copy distillable if and only if
\[
 \alpha<-\frac12.
\]
Consequently, for every $d\ge3$ and
$-\frac12\le\alpha<-1/d$, the state $\rho_\alpha$ is NPT and two-copy
undistillable.
\end{corollary}

We also discuss what this result
does, and does not, imply for $k$ copies.
\rev{Costa Rico's} partial-trace criterion \cite{CostaRico2025} associates to
$C\in\cL((\C^d)^{\otimes k})$ the form
\begin{equation}
\label{eq:intro-qk}
 q_k(\alpha,C)
 :=\sum_{J\subseteq[k]}\alpha^{\abs J}\norm{\Tr_JC}_2^2.
\end{equation}
The state $\rho_\alpha$ is $k$-copy distillable exactly when
$q_k(\alpha,C)<0$ for some $C$ of rank at most two.  At the endpoint
$\alpha=-1/2$, the unrestricted many-copy conjecture is therefore
\begin{equation}
\label{eq:intro-endpoint-conjecture}
 q_k\!\left(-\frac12,C\right)\ge0
 \quad\text{for every }k\ge1\text{ and every }\rank C\le2.
\end{equation}
For $k=2$, \cref{eq:intro-endpoint-conjecture} is exactly
\cref{eq:intro-main}.  For $k\ge3$, it is not a formal consequence of the
two-copy theorem.

Our arbitrary-copy results have three parts.  First, we give exact equivalent
forms of \cref{eq:intro-endpoint-conjecture}.  If
\begin{equation}
\label{eq:intro-Phi}
 \Phi_d(X)=\Tr(X)I_d-\frac12X,
\end{equation}
then the $k$-copy endpoint is equivalent to $2$-positivity of
$\Phi_d^{\otimes k}$.  It is also equivalent to positivity of an explicit
operator $H_k(\psi)$ for every pure state on
$U_1\otimes\cdots\otimes U_k\otimes\C^2$ whose qubit marginal is maximally
mixed.  The first three members are
\begin{align*}
 H_1(\psi)&=I-\rho,\\
 H_2(\psi)&=I-\rho_{1Q}-\rho_{2Q}+\frac12\rho,\\
 H_3(\psi)&=I-\sum_i\rho_{iQ}
 +\frac12\sum_{i<j}\rho_{ijQ}-\frac14\rho,
\end{align*}
with identity factors suppressed.  The proof below establishes $H_2\succeq0$;
the assertion $H_3\succeq0$ for all admissible $\psi$ is already equivalent
to the unresolved three-copy endpoint.

Second, we prove an exact tensorization theorem for structured witnesses.
The form $q_k$ factorizes over tensor blocks.  In any tensor-block
decomposition of a nonzero rank-at-most-two witness, at most one factor has
rank two.  We prove endpoint nonnegativity whenever that factor, if present,
is supported on a block containing at most two copies.  Consequently, a
negative witness cannot admit any tensor decomposition that localizes its
rank-two factor to such a block.

Third, we sharpen an earlier dimension-free finite-copy bound of
Qi, Gupur, Wu, and Guo \cite{QiGupurWuGuo2024}.  Their threshold is the
zero of $1+(1-a)^k-(1+a)^k$; the complement-pair argument below gives a
strictly larger threshold for every fixed $k$.  For each
$k\ge2$, let $\gamma_k$ be the unique zero in $(0,1/2]$ of the polynomial
$p_k$ defined in \cref{eq:pk-odd,eq:pk-even} below.  We prove
\begin{equation}
\label{eq:intro-gamma-result}
 \alpha\ge-\gamma_k
 \quad\Longrightarrow\quad
 \rho_\alpha\text{ is $k$-copy undistillable in every dimension}.
\end{equation}
The first values are
\[
 \gamma_2=\frac12,
 \qquad \gamma_3=\frac16,
 \qquad \gamma_4\approx0.1241036,
 \qquad \gamma_5\approx0.0981112,
\]
and
\begin{equation}
\label{eq:intro-gamma-asymptotic}
 k\gamma_k\longrightarrow
 \operatorname{arsinh}\!\left(\frac12\right)
 =\log\!\left(\frac{1+\sqrt5}{2}\right).
\end{equation}
For $k=2$, the bound is exact and recovers \cref{cor:intro-werner}.  For
$k\ge3$, it gives a rigorous extension but does not close the interval
$[-1/2,-\gamma_k)$.

In dimension four, the partial-trace theorem also yields the sharp
Kronecker-sum inequality
\begin{equation}
\label{eq:intro-kronecker}
 s_1(A\otimes I_4+I_4\otimes B)^2
 +s_2(A\otimes I_4+I_4\otimes B)^2
 \le 2\bigl(\norm{A}_2^2+\norm{B}_2^2\bigr)
\end{equation}
for traceless $A,B\in M_4(\C)$, with equality examples.

\subsection*{Provenance and concurrent work}

The discovery process was described publicly by Gajjala on X and by Bharti on
LinkedIn \cite{GajjalaX2026,BhartiLinkedIn2026}.  The posts reported an
experiment in which GPT-5.6 Sol was given basic prompts about the two-copy
Werner-state problem and returned a proof.  We record this history
because the mode of discovery is scientifically relevant.  The software is
not listed as an author.

During the same few days, Fu, Gao, and Park independently obtained the exact
two-copy threshold and posted
\emph{A solution to 2-copy distillability of Werner states} on 23 July 2026
\cite{FuGaoPark2026}.  Their argument also identifies the sharp constant
$1/2$ for a tensor antisymmetric projection.  It then proceeds through a
block-operator inequality and a variational analysis of diagonal,
off-diagonal, and transverse blocks.  The proof presented here instead uses a
dual double-skew Ky Fan inequality, co-reduction-map positivity for a
maximally mixed qubit marginal, and the standalone partial-trace inequality
\cref{eq:intro-main}.

Fraser, Huber, Pozsgay, and Vona privately communicated the concurrent
manuscript
\emph{On the two-copy distillability of Werner states} on 25 July 2026
\cite{FraserHuberPozsgayVona2026}.  According to the version communicated to
the authors, it proves the stronger rank-constrained
inequality
\[
 \norm{\Tr_U C}_2^2+\norm{\Tr_V C}_2^2
 \le r\norm{C}_2^2+\frac1r\abs{\Tr C}^2,
 \qquad \rank C\le r.
\]

We regard these works as concurrent and complementary, and we make no priority
claim beyond documenting the chronology.  The additional contribution of the
present manuscript is its systematic treatment of finite $k$: exact map and
marginal reformulations, tensorization for structured witnesses, and a
strict sharpening of the previously known dimension-free $k$-copy intervals
of \cite{QiGupurWuGuo2024}.  None of the concurrent two-copy approaches
currently settles the unrestricted endpoint for $k\ge3$.

\subsection*{Organization.}
The paper is organized as follows.  \Cref{sec:preliminaries} fixes notation
and states the arbitrary-copy rank-two criterion.  \Cref{sec:k-formulations}
derives the map and qubit-marginal formulations and isolates the obstruction
to induction.  \Cref{sec:double-skew,sec:qubit,sec:partial-trace} give the
geometric proof of the sharp two-copy inequality.
\Cref{sec:werner} derives the exact two-copy threshold.
\Cref{sec:block-tensorization} proves the structured many-copy theorem, and
\cref{sec:k-bounds} gives the explicit numbers $\gamma_k$.
\Cref{sec:kronecker} records matrix consequences before the final discussion.

\section{Preliminaries}
\label{sec:preliminaries}

All Hilbert spaces are finite-dimensional and complex.  Orthonormal bases are
fixed whenever transpose or entrywise conjugation is used.  The
Hilbert--Schmidt inner product is
\[
 \ip{X}{Y}_{\HS}:=\Tr(X^*Y),
 \qquad \norm{X}_2^2=\Tr(X^*X).
\]
The vector inner product is conjugate-linear in the first variable.
\rev{For $X\in\cL(E,F)$, the notation $\ran X:=\{Xx:x\in E\}\subseteq F$
denotes the range of $X$.}

\subsection{Vectorization and partial traces}

For $X\in\cL(V,U)$, define
\begin{equation}
\label{eq:vec}
 \operatorname{vec}(X)
 :=\sum_{i,a}X_{ia}\,u_i\otimes v_a.
\end{equation}
This is an isometry from $\cL(V,U)$ with its Hilbert--Schmidt norm onto
$U\otimes V$.  Two identities used repeatedly are
\begin{align}
 \Tr_V\bigl(
 \ketbra{\operatorname{vec}(E)}{\operatorname{vec}(D)}
 \bigr)
 &=ED^*,
 \label{eq:vec-trV}\\
 \Tr_U\bigl(
 \ketbra{\operatorname{vec}(E)}{\operatorname{vec}(D)}
 \bigr)
 &=(D^*E)^T.
 \label{eq:vec-trU}
\end{align}
Both follow by tracing the matrix coefficients over the repeated index.

If $H$ has orthonormal basis $(h_a)$, define the coefficient map
\begin{equation}
\label{eq:JH}
 J_H:\cL(H)\longrightarrow H\otimes H,
 \qquad
 J_H(X)=\sum_{a,b}X_{ab}\,h_a\otimes h_b.
\end{equation}
It is again an isometry.  \rev{For a vector in $H\otimes H$, the notation
$\SR_{H:H}$ denotes its Schmidt rank across the bipartition between the first
and second copies of $H$.}  Moreover,
\begin{equation}
\label{eq:rank-schmidt}
 \SR_{H:H}(J_H(X))=\rank X.
\end{equation}
Indeed, a singular-value decomposition of $X$ becomes a Schmidt decomposition
of $J_H(X)$ after conjugating the right singular vectors in the fixed basis.

We also use the Euclidean Ky Fan two-norm
\[
 \norm{X}_{(2),2}:=\sqrt{s_1(X)^2+s_2(X)^2}.
\]

\begin{lemma}[Rank-two duality]
\label{lem:kyfan-duality}
For every $X\in\cL(H)$,
\begin{equation}
\label{eq:kyfan-duality}
 \norm{X}_{(2),2}
 =\max_{\substack{\rank Z\le2\\ \norm{Z}_2=1}}
 \abs{\ip{Z}{X}_{\HS}}.
\end{equation}
\end{lemma}

\begin{proof}
Let $X=U\Sigma V^*$ be a singular-value decomposition.  Von Neumann's trace
inequality gives
\[
 \abs{\ip{Z}{X}_{\HS}}
 \le s_1(Z)s_1(X)+s_2(Z)s_2(X)
 \le \norm{Z}_2\norm{X}_{(2),2}
\]
for $\rank Z\le2$.  Equality is attained by taking $Z$ to have the same first
two singular directions as $X$, with singular values proportional to
$s_1(X)$ and $s_2(X)$.
\end{proof}

\subsection{The arbitrary-copy rank-two Werner criterion}
\label{subsec:qk-criterion}

Let $[k]=\{1,\ldots,k\}$ and
\[
 \mathcal H_k=H_1\otimes\cdots\otimes H_k.
\]
For $J\subseteq[k]$, let $\Tr_J$ denote partial trace over the factors indexed
by $J$, with the conventions $\Tr_\varnothing C=C$ and
$\Tr_{[k]}C=\Tr C$.  For $\alpha\in\mathbb R$, define
\begin{equation}
\label{eq:qk-def}
 q_k(\alpha,C)
 :=\sum_{J\subseteq[k]}\alpha^{\abs J}\norm{\Tr_JC}_2^2.
\end{equation}
When all $H_i=\C^d$, Pablo Costa Rico's characterization
\cite[Theorem~1]{CostaRico2025} is the following.

\begin{proposition}[Partial-trace criterion for $k$ copies]
\label{prop:costa-rico-k}
Let $d\ge2$, $k\ge1$, and let $\rho_\alpha$ be the Werner state
\cref{eq:werner-intro}.  It is $k$-copy distillable if and only if there exists
$C\in\cL((\C^d)^{\otimes k})$ with $\rank C\le2$ such that
\begin{equation}
\label{eq:qk-negative}
 q_k(\alpha,C)<0.
\end{equation}
\end{proposition}

For $k=2$, \cref{eq:qk-def} becomes
\begin{equation}
\label{eq:qalpha}
 q_2(\alpha,C)
 =\norm{C}_2^2
 +\alpha\bigl(\norm{\Tr_1C}_2^2+\norm{\Tr_2C}_2^2\bigr)
 +\alpha^2\abs{\Tr C}^2.
\end{equation}
At $\alpha=-1/2$, nonnegativity of \cref{eq:qalpha} is exactly the
partial-trace inequality \cref{eq:intro-main}.

We shall repeatedly use the elementary rank bound
\begin{equation}
\label{eq:trace-rank}
 \abs{\Tr C}^2\le \rank(C)\norm{C}_2^2.
\end{equation}
It follows from $\abs{\Tr C}\le\norm{C}_1$ and
$\norm{C}_1\le\sqrt{\rank C}\norm{C}_2$.

\section{Exact arbitrary-copy reformulations}
\label{sec:k-formulations}

The higher-copy problem can be written as tensor-power $2$-positivity or as a
hierarchy of operator inequalities with a qubit ancilla.  Both formulations
are exact; neither is an approximation.

\subsection{Tensor powers of the endpoint map}

For $\alpha\in\mathbb R$, define
\begin{equation}
\label{eq:Phi-alpha}
 \Phi_{\alpha,d}:M_d(\C)\longrightarrow M_d(\C),
 \qquad
 \Phi_{\alpha,d}(X)=\Tr(X)I_d+\alpha X.
\end{equation}
With the unnormalized maximally entangled vector
$\Omega_d=\sum_i e_i\otimes e_i$, the Choi matrix of
$\Phi_{\alpha,d}$ is
\begin{equation}
\label{eq:Phi-Choi}
 J(\Phi_{\alpha,d})=I_{d^2}+\alpha\ketbra{\Omega_d}{\Omega_d},
\end{equation}
which is a positive scalar multiple of $\rho_\alpha^\Gamma$.

\begin{theorem}[Map formulation]
\label{thm:map-formulation}
For fixed integers $d\ge2$ and $k\ge1$, and for
$\alpha\in[-1,1]$, the following are equivalent:
\begin{enumerate}
\item $\rho_\alpha$ is $k$-copy undistillable;
\item $\Phi_{\alpha,d}^{\otimes k}$ is $2$-positive;
\item $q_k(\alpha,C)\ge0$ for every
$C\in\cL((\C^d)^{\otimes k})$ with $\rank C\le2$.
\end{enumerate}
In particular, the unrestricted endpoint conjecture is equivalent to
$2$-positivity of every tensor power of
\begin{equation}
\label{eq:Phi-endpoint}
 \Phi_d:=\Phi_{-1/2,d},
 \qquad \Phi_d(X)=\Tr(X)I_d-\frac12X.
\end{equation}
\end{theorem}

\begin{proof}
A linear map is $2$-positive if and only if its Choi matrix has nonnegative
expectation on every vector of Schmidt rank at most two.  Choi matrices
multiply under tensor products after the canonical unitary regrouping from
interleaved input--output factors to all output factors followed by all input
factors.  This regrouping preserves Schmidt rank across the relevant
output--input bipartition.  Thus \cref{eq:Phi-Choi} identifies condition (2)
with block positivity of $(\rho_\alpha^\Gamma)^{\otimes k}$ on precisely the
Schmidt-rank-two test vectors in the distillability criterion.  This gives
(1)$\Leftrightarrow$(2).  The equivalence with (3) is
\cref{prop:costa-rico-k}.
\end{proof}

The phrase ``tensor-stable $2$-positivity'' is useful here: the all-copy
endpoint asks whether every $\Phi_d^{\otimes k}$ is $2$-positive.  This is
stronger than positivity of each tensor power and belongs to the general
circle of tensor-stability questions for positive maps
\cite{MullerHermesReebWolf2016,AubrunMullerHermes2021}.  There is no general
principle asserting that the tensor product of two $2$-positive maps is again
$2$-positive.

\subsection{A qubit-marginal hierarchy}

Fix $k\ge1$.  Let $U_1,\ldots,U_k$ be finite-dimensional Hilbert spaces,
let $Q=\C^2$, and
let $\psi\in U_1\otimes\cdots\otimes U_k\otimes Q$ be a unit vector with
\begin{equation}
\label{eq:k-maxmix}
 \rho_Q=\Tr_{U_1\cdots U_k}\ketbra{\psi}{\psi}=\frac12I_Q.
\end{equation}
Write $\rho=\ketbra{\psi}{\psi}$ and
$\rho_{JQ}=\Tr_{J^c}\rho$.  Identity factors on $J^c$ will be written
explicitly only when needed.  Every expression
$\rho_{JQ}\otimes I_{J^c}$ below denotes the canonical subsystem-labelled
embedding into $U_1\otimes\cdots\otimes U_k\otimes Q$, irrespective of the
displayed tensor order.  Define
\begin{equation}
\label{eq:Hk-def}
 H_k(\psi)
 :=I+\sum_{\varnothing\ne J\subseteq[k]}
 2\left(-\frac12\right)^{\abs J}
 \rho_{JQ}\otimes I_{J^c}.
\end{equation}
Equivalently,
\begin{equation}
\label{eq:Hk-map}
 H_k(\psi)
 =2\bigl(\Phi_{U_1}\otimes\cdots\otimes\Phi_{U_k}
 \otimes\id_Q\bigr)(\rho),
\end{equation}
where $\Phi_U(X)=\Tr(X)I_U-X/2$ and \cref{eq:k-maxmix} supplies the identity
term.

We first record the contraction identity behind the equivalence.

\begin{lemma}[Marginal contraction]
\label{lem:marginal-contraction}
Let $e_1,e_2$ be orthonormal in
$U_1\otimes\cdots\otimes U_k$, let $d_1,d_2$ be arbitrary, and set
\begin{equation}
\label{eq:k-C-psi-delta}
 C=\sum_{r=1}^2\ketbra{e_r}{d_r},
 \qquad
 \psi=\frac1{\sqrt2}\sum_{r=1}^2e_r\otimes q_r,
 \qquad
 \delta=\sum_{r=1}^2d_r\otimes q_r.
\end{equation}
Then, for every $J\subseteq[k]$,
\begin{equation}
\label{eq:marginal-contraction}
 \ip{\delta}{(\rho_{JQ}\otimes I_{J^c})\delta}
 =\frac12\norm{\Tr_JC}_2^2.
\end{equation}
\end{lemma}

\begin{proof}
For vectors $e,d,e',d'$ in $U_J\otimes U_{J^c}$, a coefficient expansion in
product bases gives
\begin{equation}
\label{eq:four-vector-contraction}
 \ip{d}{\bigl(\Tr_{J^c}\ketbra{e}{e'}\otimes I_{J^c}\bigr)d'}
 =\ip{\Tr_J\ketbra{e'}{d'}}
        {\Tr_J\ketbra{e}{d}}_{\HS}.
\end{equation}
Set $X_r=\Tr_J\ketbra{e_r}{d_r}$.  Expanding $\rho_{JQ}$ from
\cref{eq:k-C-psi-delta} and applying
\cref{eq:four-vector-contraction} gives
\[
 \ip{\delta}{(\rho_{JQ}\otimes I_{J^c})\delta}
 =\frac12\sum_{r,s}\ip{X_s}{X_r}_{\HS}
 =\frac12\norm{\sum_rX_r}_2^2
 =\frac12\norm{\Tr_JC}_2^2.
\]
The cases $J=\varnothing$ and $J=[k]$ respectively give
$\norm{C}_2^2/2$ and $\abs{\Tr C}^2/2$.
\end{proof}

\begin{theorem}[Qubit-marginal formulation]
\label{thm:Hk-equivalence}
For fixed $U_1,\ldots,U_k$, the following are equivalent:
\begin{enumerate}
\item $q_k(-1/2,C)\ge0$ for every
$C\in\cL(U_1\otimes\cdots\otimes U_k)$ with $\rank C\le2$;
\item $H_k(\psi)\succeq0$ for every unit vector $\psi$ satisfying
\cref{eq:k-maxmix}.
\end{enumerate}
More precisely, for $C,\psi,\delta$ as in
\cref{eq:k-C-psi-delta},
\begin{equation}
\label{eq:Hk-expectation}
 \ip{\delta}{H_k(\psi)\delta}
 =q_k\!\left(-\frac12,C\right).
\end{equation}
\end{theorem}

\begin{proof}
Let $D=\dim(U_1\otimes\cdots\otimes U_k)$.  If $D=1$, there is no
unit vector satisfying \cref{eq:k-maxmix}, while for scalar $C$ one has
\[
 q_k(-1/2,C)=2^{-k}\abs{C}^2\ge0.
\]
Thus both conditions hold.  Assume henceforth that $D\ge2$.

\rev{For $C,\psi,\delta$ as in \cref{eq:k-C-psi-delta},} insert
\cref{eq:marginal-contraction} into \cref{eq:Hk-def}.  The identity term
contributes $\norm{C}_2^2$, and every nonempty $J$ contributes
$(-1/2)^{\abs J}\norm{\Tr_JC}_2^2$, proving
\cref{eq:Hk-expectation}.  Given an arbitrary rank-at-most-two $C$, choose
orthonormal $e_1,e_2$ whose span contains $\ran C$ and set
$d_r=C^*e_r$.  Then \cref{eq:k-C-psi-delta} applies, so (2) implies (1).

Conversely, every $\psi$ satisfying \cref{eq:k-maxmix} has a Schmidt
representation of the form in \cref{eq:k-C-psi-delta} with orthonormal
$e_1,e_2$.  Every vector in the full space has the form
$\delta=d_1\otimes q_1+d_2\otimes q_2$.  The corresponding $C$ has rank at
most two, so (1) and \cref{eq:Hk-expectation} show that every quadratic form
of $H_k(\psi)$ is nonnegative.
\end{proof}

For orientation, the first three operators are
\begin{align}
 H_1(\psi)&=I-\rho,
 \label{eq:H1}\\
 H_2(\psi)&=I-\rho_{1Q}\otimes I_2-I_1\otimes\rho_{2Q}
 +\frac12\rho,
 \label{eq:H2}\\
 H_3(\psi)&=I-\sum_{i=1}^3\rho_{iQ}\otimes I_{i^c}
 +\frac12\sum_{1\le i<j\le3}\rho_{ijQ}\otimes I_{\{i,j\}^c}
 -\frac14\rho.
 \label{eq:H3}
\end{align}
The positive-operator theorem proved in \cref{sec:qubit} is precisely
$H_2(\psi)\succeq0$.  By \cref{thm:Hk-equivalence}, universal positivity of
\cref{eq:H3} would solve the three-copy endpoint.

\subsection{Why the two-copy proof does not induct}

For $S\subseteq[k]$, let
\[
 q_S(\alpha,C)=\sum_{J\subseteq S}\alpha^{\abs J}
 \norm{\Tr_JC}_2^2.
\]
Separating subsets according to whether they contain the last copy gives the
formal recursion
\begin{equation}
\label{eq:q-recursion}
 q_{[k]}(\alpha,C)
 =q_{[k-1]}(\alpha,C)
 +\alpha q_{[k-1]}(\alpha,\Tr_kC).
\end{equation}
At $\alpha=-1/2$, the second term has the unfavorable sign.  More
importantly, partial trace does not preserve the rank-two class: already a
rank-one maximally entangled projection can have a full-rank partial trace.
The first term in \cref{eq:q-recursion} also retains the last system as a
spectator.  Hence an induction that merely invokes the $(k-1)$-copy theorem
would apply that theorem outside its hypotheses.  In the map formulation,
the same obstruction is the absence of a general tensor-product closure law
for $2$-positive maps.

\section{Double-antisymmetric geometry}
\label{sec:double-skew}

For a Hilbert space $E$, let
\[
 \Alt(E):=\{L\in\cL(E):L^T=-L\}.
\]
Set $H=U\otimes V$.  After regrouping
\[
 H\otimes H
 =(U\otimes V)\otimes(U\otimes V)
 \cong(U\otimes U)\otimes(V\otimes V),
\]
let $F_U$ swap the two $U$ factors and let $F_V$ swap the two $V$ factors.
\rev{We denote the orthogonal projection onto the subspace that is
antisymmetric in both pairs by $\Paa{}$, where ``aa'' stands for
antisymmetric--antisymmetric:}
\begin{equation}
\label{eq:Paa}
 \Paa{}:=\frac14(I-F_U)(I-F_V).
\end{equation}
Its range is naturally identified with
$\bigwedge^2U\otimes\bigwedge^2V$.

\subsection{A rank-two projection estimate}

\begin{lemma}[Symmetric rank-two projection bound]
\label{lem:projection}
Let $Z\in\cL(U\otimes V)$ satisfy $Z^T=Z$ and $\rank Z\le2$.  Then
\begin{equation}
\label{eq:projection-bound}
 \norm{\Paa{}J_H(Z)}^2\le\frac12\norm{Z}_2^2.
\end{equation}
\end{lemma}

\begin{proof}
If $\dim H=1$, then $\Paa{}=0$.  Assume $\dim H\ge2$.  By Takagi
factorization \cite[Section~4.4]{HornJohnson2013},
\begin{equation}
\label{eq:takagi-Z}
 Z=a\,xx^T+b\,yy^T,
\end{equation}
where $a,b\ge0$ and $x,y\in H$ are orthonormal.  If $\rank Z=1$, take
$b=0$ and choose any unit $y\perp x$.

Write $x=\operatorname{vec}(X)$ and $y=\operatorname{vec}(Y)$ with
$X,Y\in\cL(V,U)$, and set
\begin{equation}
\label{eq:pxt}
 p_X=\Tr((X^*X)^2),
 \qquad
 p_Y=\Tr((Y^*Y)^2),
 \qquad
 t=\Tr((X^*Y)^2).
\end{equation}
Since $J_H(xx^T)=x\otimes x$, we have
\[
 z:=J_H(Z)=a\,x\otimes x+b\,y\otimes y.
\]
A direct coefficient computation gives
\begin{align}
 \ip{x\otimes x}{F_U(x\otimes x)}&=p_X,
 \label{eq:swap-xx}\\
 \ip{y\otimes y}{F_U(y\otimes y)}&=p_Y,
 \label{eq:swap-yy}\\
 \ip{x\otimes x}{F_U(y\otimes y)}&=t.
 \label{eq:swap-xy}
\end{align}
Hence
\begin{equation}
\label{eq:zFU}
 \ip{z}{F_Uz}=a^2p_X+b^2p_Y+2ab\operatorname{Re}t.
\end{equation}

We show that the right side is nonnegative.  Put $M=X^*Y$.  First,
\begin{align}
 \abs{\Tr(M^2)}
 &=\abs{\sum_{i,j}M_{ij}M_{ji}}
 \le\sum_{i,j}\abs{M_{ij}}\abs{M_{ji}}
 \le\norm{M}_2^2.
 \label{eq:trace-square}
\end{align}
Second, Hilbert--Schmidt Cauchy--Schwarz for the positive operators
$XX^*$ and $YY^*$ gives
\begin{align}
 \norm{X^*Y}_2^2
 &=\Tr(XX^*YY^*)
 \le\sqrt{\Tr((XX^*)^2)\Tr((YY^*)^2)}
 =\sqrt{p_Xp_Y}.
 \label{eq:XY-bound}
\end{align}
Therefore $\abs{t}^2\le p_Xp_Y$, so
\[
 \begin{pmatrix}p_X&t\\ \overline t&p_Y\end{pmatrix}\succeq0.
\]
\rev{Here $\overline t$ denotes the complex conjugate of the scalar $t$.}
Applying this matrix to the real vector $(a,b)^T$ and using
\cref{eq:zFU} yields
\begin{equation}
\label{eq:FU-positive}
 \ip{z}{F_Uz}\ge0.
\end{equation}

Every term $x\otimes x$ is invariant under the simultaneous exchange of the
two copies of $H$, so $F_UF_Vz=z$.  Since the swaps commute, this implies
$F_Uz=F_Vz$.  Using \cref{eq:Paa,eq:FU-positive},
\begin{align*}
 \norm{\Paa{}z}^2
 &=\ip{z}{\Paa{}z}\\
 &=\frac14\bigl(
 2\norm{z}^2-\ip{z}{F_Uz}-\ip{z}{F_Vz}
 \bigr)\\
 &=\frac12\bigl(\norm{z}^2-\ip{z}{F_Uz}\bigr)
 \le\frac12\norm{z}^2.
\end{align*}
The coefficient map $J_H$ is an isometry, which proves
\cref{eq:projection-bound}.
\end{proof}

\subsection{The sharp double-skew Ky Fan inequality}

\begin{theorem}[Double-skew Ky Fan bound]
\label{thm:double-skew}
If
\[
 K\in\Alt(U)\otimes\Alt(V)\subseteq\cL(U\otimes V),
\]
then
\begin{equation}
\label{eq:double-skew}
 s_1(K)^2+s_2(K)^2\le\frac12\norm{K}_2^2.
\end{equation}
Consequently, for every orthonormal pair $f_1,f_2\in U\otimes V$,
\begin{equation}
\label{eq:double-skew-action}
 \norm{Kf_1}^2+\norm{Kf_2}^2\le\frac12\norm{K}_2^2.
\end{equation}
If $\dim U,\dim V\ge2$, the constant $1/2$ is sharp.
\end{theorem}

\begin{proof}
For elementary tensors $A\otimes B$ with $A^T=-A$ and $B^T=-B$,
\[
 (A\otimes B)^T=A\otimes B.
\]
Thus every $K\in\Alt(U)\otimes\Alt(V)$ is complex symmetric.  Its coefficient
tensor changes sign under each of $F_U$ and $F_V$, so
\begin{equation}
\label{eq:K-range}
 K^T=K,
 \qquad
 \Paa{}J_H(K)=J_H(K).
\end{equation}
Take a Takagi factorization
\[
 K=\sum_j s_j(K)u_ju_j^T
\]
with orthonormal $u_j$.  Let
\[
 s=\sqrt{s_1(K)^2+s_2(K)^2}.
\]
If $s=0$, there is nothing to prove.  Otherwise set
\[
 Z=\frac{s_1(K)u_1u_1^T+s_2(K)u_2u_2^T}{s}.
\]
Then $Z^T=Z$, $\rank Z\le2$, and $\norm{Z}_2=1$.  Orthogonality of the
Takagi terms gives
\[
 s=\ip{Z}{K}_{\HS}
 =\ip{\Paa{}J_H(Z)}{J_H(K)}.
\]
By \cref{lem:projection} and Cauchy--Schwarz,
\[
 s\le\norm{\Paa{}J_H(Z)}\norm{K}_2
 \le\frac1{\sqrt2}\norm{K}_2,
\]
which proves \cref{eq:double-skew}.  The variational principle for the two
largest eigenvalues of $K^*K$ gives
\[
 \norm{Kf_1}^2+\norm{Kf_2}^2
 \le s_1(K)^2+s_2(K)^2,
\]
and hence \cref{eq:double-skew-action}.

For sharpness, choose the spans of two vectors from each of the fixed
orthonormal bases of $U$ and $V$, and let
\[
 L=\begin{pmatrix}0&1\\-1&0\end{pmatrix}
\]
act on each.  Extended by zero on the orthogonal complements,
$K=L\otimes L$ has four nonzero singular values, all equal to one.  Thus
$s_1(K)^2+s_2(K)^2=2=\frac12\norm{K}_2^2$.
\end{proof}

The preceding theorem has an equivalent Schmidt-rank formulation.  For a
positive operator $Q$ on $H\otimes H$, define
\begin{equation}
\label{eq:S2-def}
 \norm{Q}_{S(2)}
 :=\sup_{\substack{\norm{\psi}=1\\ \SR_{H:H}(\psi)\le2}}
 \ip{\psi}{Q\psi}.
\end{equation}

\begin{corollary}[Sharp antisymmetric-projection norm]
\label{cor:S2}
Assume $\dim U,\dim V\ge2$.  For $\Paa{}$ in \cref{eq:Paa},
\begin{equation}
\label{eq:S2-half}
 \norm{\Paa{}}_{S(2)}=\frac12.
\end{equation}
\end{corollary}

\begin{proof}
Because $\Paa{}$ is an orthogonal projection,
\[
 \ip{\psi}{\Paa{}\psi}
 =\max_{\substack{\varphi\in\ran \Paa{}\\ \norm{\varphi}=1}}
 \abs{\ip{\psi}{\varphi}}^2.
\]
Taking the supremum first over Schmidt-rank-two unit vectors and then over
unit $\varphi\in\ran \Paa{}$ gives
\begin{align*}
 \norm{\Paa{}}_{S(2)}
 &=\sup_{\substack{K\in\Alt(U)\otimes\Alt(V)\\\norm{K}_2=1}}
 \sup_{\substack{\rank Z\le2\\\norm{Z}_2=1}}
 \abs{\ip{Z}{K}_{\HS}}^2\\
 &=\sup_{\substack{K\in\Alt(U)\otimes\Alt(V)\\\norm{K}_2=1}}
 \bigl(s_1(K)^2+s_2(K)^2\bigr),
\end{align*}
where we used \cref{eq:rank-schmidt,lem:kyfan-duality}.  The upper bound is
\cref{thm:double-skew}, and the sharpness example in that theorem gives the
matching lower bound.
\end{proof}

\begin{remark}
\Cref{cor:S2} is the common geometric constant that also appears in the
concurrent proof of Fu, Gao, and Park \cite{FuGaoPark2026}.  In the present
argument it is used through the dual action estimate
\cref{eq:double-skew-action}.
\end{remark}

\section{A positive operator with a maximally mixed qubit marginal}
\label{sec:qubit}

The double-skew estimate controls the negative eigenvector of a partially
transposed Schmidt-rank-two pure state.  The result is the following
tripartite operator inequality.

\subsection{Reduction and co-reduction maps}

For a finite-dimensional Hilbert space $E$, define
\begin{equation}
\label{eq:reduction-maps}
 \cR_E(X)=\Tr(X)I_E-X,
 \qquad
 \Lambda_E(X)=\Tr(X)I_E-X^T.
\end{equation}
Let $(e_i)$ be the fixed basis and, for $i<j$, set
\begin{equation}
\label{eq:skew-kraus}
 L_{ij}=\ketbra{e_i}{e_j}-\ketbra{e_j}{e_i}.
\end{equation}
A direct calculation gives
\begin{equation}
\label{eq:lambda-kraus}
 \Lambda_E(X)=\sum_{i<j}L_{ij}XL_{ij}^*.
\end{equation}
Thus $\Lambda_E$ is completely positive and
\begin{equation}
\label{eq:R-Lambda}
 \cR_E=\Lambda_E\circ T.
\end{equation}

\begin{proposition}[Qubit-marginal operator inequality]
\label{prop:qubit}
Let $Q=\C^2$, and let $\psi\in U\otimes V\otimes Q$ be a unit vector such
that
\begin{equation}
\label{eq:qubit-maxmix}
 \rho_Q:=\Tr_{UV}\ketbra{\psi}{\psi}=\frac12I_Q.
\end{equation}
Write
\[
 \rho=\ketbra{\psi}{\psi},
 \qquad
 \rho_{UQ}=\Tr_V\rho,
 \qquad
 \rho_{VQ}=\Tr_U\rho.
\]
Then
\begin{equation}
\label{eq:H-positive}
 H_\psi
 :=I_{UVQ}+\frac12\rho
 -\rho_{UQ}\otimes I_V-I_U\otimes\rho_{VQ}
 \succeq0,
\end{equation}
where identity factors are placed in their natural tensor positions.
\end{proposition}

\begin{proof}
Define the completely positive map
\[
 \Phi=\Lambda_U\otimes\Lambda_V\otimes\id_Q.
\]
Using \cref{eq:R-Lambda} and expanding the two reduction maps,
\begin{align}
 \Phi(\rho^{T_{UV}})
 &=(\cR_U\otimes\cR_V\otimes\id_Q)(\rho)\notag\\
 &=\frac12I_{UVQ}
 -\rho_{UQ}\otimes I_V-I_U\otimes\rho_{VQ}+\rho,
 \label{eq:Phi-expanded}
\end{align}
where \cref{eq:qubit-maxmix} was used in the first term.  Hence
\begin{equation}
\label{eq:H-first-decomp}
 H_\psi=\Phi(\rho^{T_{UV}})+\frac12(I_{UVQ}-\rho).
\end{equation}

The maximally mixed qubit marginal implies the Schmidt form
\begin{equation}
\label{eq:psi-schmidt}
 \psi=\frac1{\sqrt2}(e_1\otimes q_1+e_2\otimes q_2),
\end{equation}
where $e_1,e_2\in U\otimes V$ and $q_1,q_2\in Q$ are orthonormal.  Define
\begin{equation}
\label{eq:eta}
 \eta=\frac1{\sqrt2}
 (\overline e_1\otimes q_2-\overline e_2\otimes q_1),
 \qquad
 P_-=\ketbra{\eta}{\eta}.
\end{equation}
Let $P_+$ be the orthogonal projection onto the span of
\[
 \overline e_1\otimes q_1,
 \quad
 \overline e_2\otimes q_2,
 \quad
 \frac1{\sqrt2}(\overline e_1\otimes q_2+
 \overline e_2\otimes q_1).
\]
On this four-dimensional Schmidt support, partial transpose on $UV$ has three
positive eigenvalues $1/2$ and one negative eigenvalue $-1/2$.  Therefore
\begin{equation}
\label{eq:pt-spectrum}
 \rho^{T_{UV}}=\frac12P_+-\frac12P_-.
\end{equation}
Substituting into \cref{eq:H-first-decomp},
\begin{equation}
\label{eq:H-second-decomp}
 H_\psi
 =\frac12\Phi(P_+)
 +\frac12\bigl(I_{UVQ}-\rho-\Phi(P_-)\bigr).
\end{equation}
The first summand is positive.  It remains to prove
\begin{equation}
\label{eq:negative-control}
 \Phi(P_-)\preceq I_{UVQ}-\rho.
\end{equation}

Let $(L_\alpha)$ and $(M_\beta)$ be the skew Kraus operators
\cref{eq:skew-kraus} for $U$ and $V$.  On the coefficient Hilbert space with
standard basis indexed by $(\alpha,\beta)$, define
\begin{equation}
\label{eq:T-map}
 \mathcal T c
 =\sum_{\alpha,\beta}c_{\alpha\beta}
 (L_\alpha\otimes M_\beta\otimes I_Q)\eta.
\end{equation}
Then
\begin{equation}
\label{eq:TTstar}
 \Phi(P_-)=\mathcal T\mathcal T^*.
\end{equation}
For a coefficient array $c$, set
\begin{equation}
\label{eq:Kc}
 K_c=\sum_{\alpha,\beta}
 c_{\alpha\beta}L_\alpha\otimes M_\beta.
\end{equation}
Using \cref{eq:eta},
\begin{equation}
\label{eq:Tc-norm}
 \norm{\mathcal T c}^2
 =\frac12\bigl(
 \norm{K_c\overline e_1}^2+\norm{K_c\overline e_2}^2
 \bigr).
\end{equation}
The skew matrix units are mutually Hilbert--Schmidt orthogonal and each has
squared norm $2$.  Hence
\begin{equation}
\label{eq:Kc-norm}
 \norm{K_c}_2^2=4\norm{c}_2^2.
\end{equation}
Since $K_c\in\Alt(U)\otimes\Alt(V)$, \cref{eq:double-skew-action} applied to
$\overline e_1,\overline e_2$ gives
\[
 \norm{\mathcal T c}^2
 \le\frac14\norm{K_c}_2^2
 =\norm{c}_2^2.
\]
Thus $\norm{\mathcal T}\le1$, and
\begin{equation}
\label{eq:Phi-contraction}
 0\preceq\Phi(P_-)\preceq I_{UVQ}.
\end{equation}

It remains to locate the support.  Every $K_c$ is complex symmetric because
it is a sum of tensor products of two skew-symmetric matrices.  From
\cref{eq:psi-schmidt,eq:eta},
\begin{align}
 \ip{\psi}{\mathcal T c}
 &=\frac12\bigl(
 \ip{e_2}{K_c\overline e_1}
 -\ip{e_1}{K_c\overline e_2}
 \bigr)\notag\\
 &=\frac12\bigl(
 \overline e_2^{\,T}K_c\overline e_1
 -\overline e_1^{\,T}K_c\overline e_2
 \bigr)=0.
 \label{eq:T-range}
\end{align}
Therefore $\ran\mathcal T\subseteq\psi^\perp$.  A positive contraction
supported on $\psi^\perp$ is bounded above by the projection onto
$\psi^\perp$, and hence
\[
 \mathcal T\mathcal T^*\preceq I_{UVQ}-\rho.
\]
This proves \cref{eq:negative-control}; then \cref{eq:H-second-decomp} yields
$H_\psi\succeq0$.
\end{proof}

\section{The sharp rank-two partial-trace inequality}
\label{sec:partial-trace}

We now turn the tripartite positivity statement into the desired matrix
inequality.

\begin{theorem}
\label{thm:main}
Let $C\in\cL(U\otimes V)$ satisfy $\rank C\le2$.  Then
\begin{equation}
\label{eq:main}
 \norm{\Tr_U C}_2^2+\norm{\Tr_V C}_2^2
 \le 2\norm{C}_2^2+\frac12\abs{\Tr C}^2.
\end{equation}
If $\dim U,\dim V\ge2$, the constants are jointly sharp.
\end{theorem}

\begin{proof}
If $\dim(U\otimes V)=1$, then $C$ is scalar, and the two sides of
\cref{eq:main} are $2\abs{C}^2$ and $\frac52\abs{C}^2$, respectively.
Assume $\dim(U\otimes V)\ge2$.  Choose orthonormal
$e_1,e_2\in U\otimes V$ whose span contains $\ran C$ and set
$d_r=C^*e_r$.  Then
$C=\sum_{r=1}^2\ketbra{e_r}{d_r}$ because the orthogonal projection onto
$\operatorname{span}\{e_1,e_2\}$ fixes $\ran C$.  Form $\psi,\delta$ as in
\cref{eq:k-C-psi-delta} with $k=2$, $U_1=U$, and $U_2=V$.  By
\cref{eq:Hk-expectation},
\[
 q_2\!\left(-\frac12,C\right)
 =\ip{\delta}{H_2(\psi)\delta}.
\]
The qubit marginal of $\psi$ is $I_Q/2$, and the operator $H_2(\psi)$ in
\cref{eq:H2} is exactly $H_\psi$ in \cref{eq:H-positive}.  Hence
\cref{prop:qubit} gives $q_2(-1/2,C)\ge0$.  Expanding this inequality with
\cref{eq:qalpha} gives
\[
 0\le\norm{C}_2^2
 -\frac12\bigl(\norm{\Tr_UC}_2^2+\norm{\Tr_VC}_2^2\bigr)
 +\frac14\abs{\Tr C}^2,
\]
which rearranges to \cref{eq:main}.

For sharpness, assume $\dim U,\dim V\ge2$.  Choose orthonormal
$u_0,u_1\in U$ and set
\[
 P=\ketbra{u_0}{u_0}+\ketbra{u_1}{u_1}.
\]
For unit $v,w\in V$, let
\begin{equation}
\label{eq:sharp-C}
 C=P\otimes\ketbra{v}{w},
 \qquad
 \tau=\ip{w}{v}.
\end{equation}
Then
\[
 \rank C=2,
 \qquad
 \norm{C}_2^2=2,
 \qquad
 \abs{\Tr C}^2=4\abs{\tau}^2,
\]
and
\[
 \norm{\Tr_UC}_2^2+\norm{\Tr_VC}_2^2
 =4+2\abs{\tau}^2.
\]
Thus equality holds in \cref{eq:main} for every attainable $\tau$.  Taking
$v\perp w$ shows that the coefficient $2$ cannot be lowered, even under the
constraint $\Tr C=0$.  Taking $\tau\ne0$ then shows that, once the coefficient
$2$ is fixed, the coefficient $1/2$ cannot be lowered.
\end{proof}

\begin{remark}[Basis independence]
The proof uses transpose and coordinatewise conjugation after fixing bases,
but the conclusion \cref{eq:main} is invariant under independent unitary
changes of basis on $U$ and $V$.
\end{remark}

A useful traceless version follows by removing the scalar components of both
partial traces.  Let $m=\dim U$ and $n=\dim V$, and define
\[
 \Pi_U^0(X)=X-\frac{\Tr X}{m}I_U,
 \qquad
 \Pi_V^0(Y)=Y-\frac{\Tr Y}{n}I_V.
\]

\begin{corollary}[Sharp traceless partial-trace inequality]
\label{cor:traceless-ptr}
For every rank-at-most-two $C\in\cL(U\otimes V)$,
\begin{align}
 &\norm{\Pi_V^0(\Tr_UC)}_2^2
 +\norm{\Pi_U^0(\Tr_VC)}_2^2\notag\\
 &\qquad\le
 2\norm{C}_2^2
 +\left(\frac12-\frac1m-\frac1n\right)\abs{\Tr C}^2.
 \label{eq:traceless-ptr}
\end{align}
If $m,n\ge2$, the two coefficients are jointly sharp.
\end{corollary}

\begin{proof}
Orthogonal projection onto the traceless subspace gives
\[
 \norm{\Pi_V^0(\Tr_UC)}_2^2
 =\norm{\Tr_UC}_2^2-\frac1n\abs{\Tr C}^2,
\]
and similarly for the other subsystem.  Subtracting these scalar
contributions from \cref{eq:main} proves the inequality.

For sharpness, use $C$ from \cref{eq:sharp-C}.  Since
$\Tr_UC=2\ketbra{v}{w}$ and $\Tr_VC=\tau P$, one obtains
\begin{align*}
 \norm{\Pi_V^0(\Tr_UC)}_2^2
 &=4-\frac4n\abs{\tau}^2,\\
 \norm{\Pi_U^0(\Tr_VC)}_2^2
 &=\left(2-\frac4m\right)\abs{\tau}^2.
\end{align*}
Their sum is
\[
 4+\left(2-\frac4m-\frac4n\right)\abs{\tau}^2,
\]
which equals the right side of \cref{eq:traceless-ptr}.  Taking
$\tau=0$ fixes the coefficient $2$, and then taking $\tau\ne0$ fixes the
trace coefficient.
\end{proof}

\section{The exact two-copy Werner threshold}
\label{sec:werner}

We now combine \cref{thm:main} with the criterion
\cref{prop:costa-rico-k}.

\begin{theorem}
\label{thm:werner-threshold}
For every $d\ge2$, the Werner state $\rho_\alpha$ is two-copy undistillable if
and only if
\begin{equation}
\label{eq:werner-threshold}
 -\frac12\le\alpha\le1.
\end{equation}
Equivalently, its one-copy and two-copy distillability thresholds coincide.
\end{theorem}

\begin{proof}
Let $C$ have rank at most two, and abbreviate
\[
 N=\norm{C}_2^2,
 \qquad
 S=\norm{\Tr_1C}_2^2+\norm{\Tr_2C}_2^2,
 \qquad
 T=\abs{\Tr C}^2.
\]
By \cref{thm:main},
\begin{equation}
\label{eq:S-bound}
 S\le2N+\frac12T.
\end{equation}
For $\alpha\ge0$, all terms in $\rev{q_2(\alpha,C)}=N+\alpha S+\alpha^2T$ are
nonnegative.  Now suppose $-1/2\le\alpha\le0$.  Since $\alpha\le0$,
\cref{eq:S-bound} gives
\begin{align}
 \rev{q_2(\alpha,C)}
 &\ge N+\alpha\left(2N+\frac12T\right)+\alpha^2T\notag\\
 &=(1+2\alpha)N+\alpha\left(\alpha+\frac12\right)T.
 \label{eq:q-lower-one}
\end{align}
The coefficient of $T$ is nonpositive.  By \cref{eq:trace-rank},
$T\le2N$, and therefore
\begin{align}
 \rev{q_2(\alpha,C)}
 &\ge(1+2\alpha)N
 +2\alpha\left(\alpha+\frac12\right)N\notag\\
 &=(1+\alpha)(1+2\alpha)N\ge0.
 \label{eq:q-lower-two}
\end{align}
Thus no rank-two witness exists for $\alpha\ge-1/2$, so
\cref{prop:costa-rico-k} implies two-copy undistillability.

Conversely, let $\alpha<-1/2$.  Choose a rank-two orthogonal projection
$P\in M_d(\C)$ and unit vectors $v,w\in\C^d$ with $v\perp w$.  Set
\[
 C=P\otimes\ketbra{v}{w}.
\]
Then $\rank C=2$, $\norm{C}_2^2=2$, the two partial traces have
squared Hilbert--Schmidt norms $4$ and $0$, and $\Tr C=0$.  Hence
\[
 \rev{q_2(\alpha,C)}=2+4\alpha<0.
\]
By \cref{prop:costa-rico-k}, $\rho_\alpha$ is two-copy distillable.
\end{proof}

\begin{corollary}
\label{cor:NPT-two-copy}
For $d\ge3$ and $-1/2\le\alpha<-1/d$, the Werner state $\rho_\alpha$ is NPT
and two-copy undistillable.  If it is distillable at all, at least three copies
are required.
\end{corollary}

\begin{remark}[The endpoint beyond two copies]
Two-copy undistillability is not the same as bound entanglement.  By
\cref{thm:map-formulation,thm:Hk-equivalence}, the higher-copy endpoint asks
for tensor-stable $2$-positivity of $\Phi_d$ or, equivalently, positivity of
all operators $H_k(\psi)$.  The next two sections establish this only for
structured witnesses and for a dimension-free interval of parameters.
\end{remark}

\section{Exact tensorization for structured witnesses}
\label{sec:block-tensorization}

Although unrestricted $2$-positivity is not stable under tensor products, the
partial-trace form itself factorizes exactly on tensor-factorized operators.
This gives a broad class of witnesses for which the endpoint follows from the
one- and two-copy theorems.

\begin{lemma}[Rank-one witnesses]
\label{lem:rank-one-all-k}
For every $k\ge1$, every $\alpha\in[-1,1]$, and every operator $C$ of rank at
most one,
\begin{equation}
\label{eq:rank-one-q-positive}
 q_k(\alpha,C)\ge0.
\end{equation}
\end{lemma}

\begin{proof}
For each $i$, let $F_i$ be the flip and let $\Omega_i$ be the unnormalized
maximally entangled vector on $H_i\otimes H_i$.  Since
\[
 I+\alpha F_i\succeq0\qquad(-1\le\alpha\le1),
\]
its partial transpose
$G_{\alpha,i}=I+\alpha\ketbra{\Omega_i}{\Omega_i}$ is block positive.
After canonically regrouping all first tensor factors against all second
tensor factors, $\bigotimes_iG_{\alpha,i}$ is the partial transpose of the
positive operator $\bigotimes_i(I+\alpha F_i)$ and is therefore block
positive across that bipartition.  Expanding the product and using the
coefficient map \cref{eq:JH} gives
\[
 \ip{J_{\mathcal H_k}(C)}
 {\left(\bigotimes_{i=1}^kG_{\alpha,i}\right)J_{\mathcal H_k}(C)}
 =\sum_{J\subseteq[k]}\alpha^{\abs J}\norm{\Tr_JC}_2^2
 =q_k(\alpha,C).
\]
If $\rank C\le1$, then \cref{eq:rank-schmidt} makes
$J_{\mathcal H_k}(C)$ a product vector (or zero), so block positivity proves
the claim.  This argument requires neither normalization nor equal
dimensions, and it also covers one-dimensional factors.
\end{proof}

\begin{proposition}[Exact block factorization]
\label{prop:q-factorization}
Let $\mathcal P=\{B_1,\ldots,B_m\}$ be a partition of $[k]$, and suppose
\begin{equation}
\label{eq:C-block-factor}
 C=C_{B_1}\otimes\cdots\otimes C_{B_m},
 \qquad
 C_{B_r}\in\cL\!\left(\bigotimes_{i\in B_r}H_i\right).
\end{equation}
Then
\begin{equation}
\label{eq:q-block-factor}
 q_k(\alpha,C)
 =\prod_{r=1}^m q_{\abs{B_r}}(\alpha,C_{B_r}).
\end{equation}
\end{proposition}

\begin{proof}
For every $J\subseteq[k]$,
\[
 \Tr_JC
 =\bigotimes_{r=1}^m\Tr_{J\cap B_r}C_{B_r},
\]
so the squared Hilbert--Schmidt norm is the product of the corresponding
block norms.  Summing over all $J$ factors into the product of the $m$ power
sets, which is \cref{eq:q-block-factor}.
\end{proof}

\begin{theorem}[Rank-two cores on at most two copy blocks]
\label{thm:block-core}
Let $C$ be nonzero with $\rank C\le2$ and suppose it admits a factorization
\cref{eq:C-block-factor}.  If every factor of rank two is supported on a
block of cardinality at most two, then
\begin{equation}
\label{eq:block-core-positive}
 q_k(\alpha,C)\ge0
 \qquad\text{for every }\alpha\in[-1/2,1].
\end{equation}
In particular, such a $C$ cannot witness $k$-copy distillability at the
endpoint.
\end{theorem}

\begin{proof}
Ranks multiply under tensor products.  Since their product is at most two,
all block factors have rank one except possibly one factor of rank two.  By
\cref{lem:rank-one-all-k}, every rank-one block contributes a nonnegative
factor to \cref{eq:q-block-factor}.  A possible rank-two factor lives on one
or two copies.  On one copy,
\[
 q_1(\alpha,D)=\norm{D}_2^2+\alpha\abs{\Tr D}^2\ge0
 \qquad(\alpha\ge-1/2),
\]
by \cref{eq:trace-rank}.  On two copies, the dimension-free
\cref{thm:main} and the argument in
\cref{eq:q-lower-one,eq:q-lower-two} give
$q_2(\alpha,D)\ge0$ for $\alpha\in[-1/2,0]$, while the assertion is
immediate for $\alpha\ge0$.  Thus every factor in
\cref{eq:q-block-factor} is nonnegative.
\end{proof}

\begin{corollary}[Copy-irreducibility of a possible counterexample]
\label{cor:copy-irreducible}
If $q_k(-1/2,C)<0$ for some rank-two $C$, then no tensor decomposition of
$C$ consists entirely of rank-one factors or localizes its unique possible
rank-two factor to one or two copy blocks.
\end{corollary}

The theorem does not assume that the rank-one blocks are single-copy
operators; they may themselves act on arbitrarily many copies.  What is
excluded is precisely a witness obtained by padding a one- or two-copy
rank-two obstruction with rank-one tensor factors.  No canonical finest
tensor factorization is asserted.

\section{Dimension-free bounds for an arbitrary finite number of copies}
\label{sec:k-bounds}

We next obtain an explicit interval of $k$-copy undistillability.  The argument
combines two complementary bipartite inequalities.  For
$C\in\cL(H_J\otimes H_{J^c})$ of rank at most two, set
\begin{equation}
\label{eq:NTA}
 N=\norm{C}_2^2,
 \qquad T=\abs{\Tr C}^2,
 \qquad A_J=\norm{\Tr_JC}_2^2.
\end{equation}
The sharp sum inequality proved here and the difference inequality of Costa
Rico \cite[Theorem~2]{CostaRico2025} give
\begin{align}
 A_J+A_{J^c}&\le2N+\frac12T,
 \label{eq:complement-sum}\\
 \abs{A_J-A_{J^c}}&\le2N-\frac12T.
 \label{eq:complement-difference}
\end{align}
For rank two, the second estimate is exactly the rank term in the cited
theorem.  For rank one, that theorem gives the stronger bound
$\abs{A_J-A_{J^c}}\le N-T$, which implies
\cref{eq:complement-difference}; the zero operator is immediate.  Thus both
displayed estimates hold uniformly on the rank-at-most-two class.  Two
consequences will be used
repeatedly:
\begin{align}
 A_J&\le2N,
 \label{eq:single-marginal-2N}\\
 wA_J+vA_{J^c}&\le2wN+\frac v2T
 \quad(0\le v\le w).
 \label{eq:weighted-complement}
\end{align}
The second follows by writing the left side as a positive combination of
$A_J+A_{J^c}$ and $A_J-A_{J^c}$ and applying
\cref{eq:complement-sum,eq:complement-difference}.

For $a\in[0,1/2]$, define the following polynomials.  If $k=2m+1\ge3$ is odd,
set
\begin{equation}
\label{eq:pk-odd}
 p_{2m+1}(a)
 :=1-2\sum_{\ell=0}^{m-1}
 \binom{2m+1}{2\ell+1}a^{2\ell+1}.
\end{equation}
If $k=2m\ge2$ is even, set
\begin{equation}
\label{eq:pk-even}
\begin{split}
 p_{2m}(a)
 := {}&1+2a^{2m}
 -\sum_{\substack{1\le j<m\\ j\ \mathrm{odd}}}
 \binom{2m}{j}\bigl(2a^j+a^{2m-j}\bigr)\\
 &-\mathbf 1_{\{m\ \mathrm{odd}\}}
 \frac32\binom{2m}{m}a^m.
\end{split}
\end{equation}
For example,
\begin{align*}
 p_2(a)&=1-3a+2a^2=(1-a)(1-2a),\\
 p_3(a)&=1-6a,\\
 p_4(a)&=1-8a-4a^3+2a^4,\\
 p_5(a)&=1-10a-20a^3.
\end{align*}

There is a closely related earlier dimension-free bound.  Define
\[
 r_k(a):=1+(1-a)^k-(1+a)^k.
\]
Qi, Gupur, Wu, and Guo proved that if $\eta_k$ is the zero of $r_k$ in
$(0,1]$, then $\alpha\ge-\eta_k$ implies $k$-copy undistillability
\cite[Theorem~IV.2]{QiGupurWuGuo2024}.  The present polynomials give a
strict finite-$k$ improvement.  Indeed, direct comparison with
\cref{eq:pk-odd,eq:pk-even} yields
\begin{align}
 p_k(a)-r_k(a)
 &=
 \begin{cases}
  2a^k, & k\ \mathrm{odd},\\[3pt]
  2a^k+
  \displaystyle\sum_{\substack{1\le j<k/2\\j\ \mathrm{odd}}}
  \binom{k}{j}a^{k-j}
  +\displaystyle\mathbf1_{\{k/2\ \mathrm{odd}\}}
  \frac12\binom{k}{k/2}a^{k/2},
  & k\ \mathrm{even}.
 \end{cases}
 \label{eq:pk-prior-comparison}
\end{align}
Hence $p_k(a)>r_k(a)$ for $a>0$, and consequently
$\gamma_k>\eta_k$.  The two families nevertheless have the same leading
$1/k$ asymptotic scale, as made explicit below.

\begin{theorem}[Complement-pair lower bound]
\label{thm:pk-lower}
Let $k\ge2$, let $0\le a\le1/2$, and let
$C\in\cL(H_1\otimes\cdots\otimes H_k)$ have rank at most two.  Then
\begin{equation}
\label{eq:pk-lower}
 q_k(-a,C)\ge p_k(a)\norm{C}_2^2.
\end{equation}
\end{theorem}

\begin{proof}
Assume $C\ne0$ and use the notation \cref{eq:NTA}.  We pair each subset
$J\subseteq[k]$ with its complement.

Suppose first that $k=2m+1$ is odd.  Keep the complementary pair
$(\varnothing,[k])$ exact; it contributes $N-a^kT$.  For every remaining
pair choose the representative with $1\le j=\abs J\le m$.  The two
coefficients have opposite signs.  If $j$ is odd, then
$a^j\ge a^{k-j}$ and
\cref{eq:complement-sum,eq:complement-difference} give
\begin{equation}
\label{eq:odd-pair-one}
 -a^jA_J+a^{k-j}A_{J^c}
 \ge-2a^jN+\frac12a^{k-j}T.
\end{equation}
\begingroup
\color{revisioncyan}
Indeed, for $0\le v\le w$,
\[
 -wA_J+vA_{J^c}
 =-\frac{w-v}{2}(A_J+A_{J^c})
  -\frac{w+v}{2}(A_J-A_{J^c})
 \ge-2wN+\frac v2T,
\]
where the last step uses
\cref{eq:complement-sum,eq:complement-difference}.  Taking
$w=a^j$ and $v=a^{k-j}$ proves \cref{eq:odd-pair-one}.
\endgroup
If $j$ is even (hence $j\ge2$), then
\begin{equation}
\label{eq:odd-pair-two}
 a^jA_J-a^{k-j}A_{J^c}
 \ge-2a^{k-j}N+\frac12a^{k-j}T,
\end{equation}
which follows from
$A_{J^c}-A_J\le2N-T/2$ and $a^j\ge a^{k-j}$.
Together, with the binomial multiplicities included, the negative
$N$-terms in \cref{eq:odd-pair-one,eq:odd-pair-two} run once through all odd
cardinalities $1,3,\ldots,k-2$.  The coefficient of $T$, including the exact
full-set term $-a^kT$, is nonnegative, because the $k$ pairs represented by
singletons already contribute $k a^{k-1}T/2$.  Thus it is at least
\[
 -a^k+\frac{k}{2}a^{k-1}
 =a^{k-1}\left(\frac{k}{2}-a\right)\ge0.
\]
Discarding that nonnegative contribution gives
\cref{eq:pk-lower} with \cref{eq:pk-odd}.

Now let $k=2m$ be even.  Keep the pair
$(\varnothing,[k])$ exact; it contributes $N+a^{2m}T$.  All other
complementary subsets have the same parity.  Pairs of positive even
cardinality contribute nonnegative terms and may be discarded.  If $j<m$ is
odd, then \cref{eq:weighted-complement} gives
\begin{equation}
\label{eq:even-pair}
 -a^jA_J-a^{k-j}A_{J^c}
 \ge-2a^jN-\frac12a^{k-j}T.
\end{equation}
If $m$ is odd, the middle layer $\abs J=m$ consists of
$\binom{2m}{m}/2$ complementary pairs, and
\cref{eq:complement-sum} bounds their total contribution below by
\[
 -\frac12\binom{2m}{m}a^m
 \left(2N+\frac12T\right).
\]
Combining these estimates with the exact empty--full pair gives an affine
lower bound in $T$, namely
\begin{align*}
 q_{2m}(-a,C)
 \ge{}&
 \left[
 1-2\sum_{\substack{1\le j<m\\j\ \mathrm{odd}}}
 \binom{2m}{j}a^j
 -\mathbf1_{\{m\ \mathrm{odd}\}}\binom{2m}{m}a^m
 \right]N\\
 &+\left[
 a^{2m}
 -\frac12\sum_{\substack{1\le j<m\\j\ \mathrm{odd}}}
 \binom{2m}{j}a^{2m-j}
 -\mathbf1_{\{m\ \mathrm{odd}\}}
 \frac14\binom{2m}{m}a^m
 \right]T.
\end{align*}
The $T$-coefficient is nonpositive for
$0\le a\le1/2$: if $m$ is even, the
$j=1$ term in \cref{eq:even-pair} already gives
$a^{2m}-(2m/2)a^{2m-1}\le0$; if $m$ is odd, the middle-layer term gives
$a^{2m}-\binom{2m}{m}a^m/4\le0$.
\rev{For the latter inequality, $a^m\le2^{-m}\le\binom{2m}{m}/4$.}
We may therefore use
$T\le2N$, which is \cref{eq:trace-rank}.  The resulting coefficient of $N$
is exactly \cref{eq:pk-even}.
\end{proof}

\begin{definition}
\label{def:gamma-k}
For $k\ge2$, let $\gamma_k$ be the unique zero of $p_k$ in $(0,1/2]$.
\end{definition}

The zero is well-defined.  For odd $k$, every nonconstant term in
$p_k$ has a negative coefficient, so $p_k$ is strictly decreasing.  For even
$k\ge4$, the $j=1$ term is present and
\[
 p_k'(a)\le 2ka^{k-1}-2k<0
 \qquad (0<a\le1/2),
\]
because all remaining derivatives are nonpositive.  The case $k=2$ follows
from $p_2(a)=(1-a)(1-2a)$.  Moreover, $p_k(0)=1$,
$p_2(1/2)=0$, and $p_k(1/2)<0$ for $k\ge3$: for odd $k$ the singleton term
alone gives $p_k(1/2)\le1-k$, while for even $k\ge4$ the singleton term gives
\[
 p_k(1/2)\le
 1+2^{1-k}-k-k2^{1-k}<0.
\]
Thus there is exactly one zero in $(0,1/2]$.

\begin{corollary}[Dimension-free $k$-copy interval]
\label{cor:gamma-k}
For every $d\ge2$, every $k\ge2$, and every $\alpha\in[-1,1]$,
\begin{equation}
\label{eq:gamma-interval}
 \alpha\ge-\gamma_k
 \quad\Longrightarrow\quad
 \rho_\alpha^{(d)}\text{ is $k$-copy undistillable}.
\end{equation}
On the other hand, $\rho_\alpha^{(d)}$ is $k$-copy distillable whenever
$\alpha<-1/2$.
\end{corollary}

\begin{proof}
For $-\gamma_k\le\alpha\le0$, write $a=-\alpha$ and apply
\cref{thm:pk-lower}; for $\alpha\ge0$, every term in
\cref{eq:qk-def} is nonnegative.  The criterion
\cref{prop:costa-rico-k} proves undistillability.  If $\alpha<-1/2$, the state
is already one-copy distillable, and extra copies may be discarded.
\end{proof}

The first values are shown in \cref{tab:gamma}.
\begin{table}[H]
\centering
\caption{\rev{The earlier constants $\eta_k$ of
\cite{QiGupurWuGuo2024} and the sharpened complement-pair constants
$\gamma_k$.  Every entry was independently recomputed as the unique positive
zero of $r_k$ or $p_k$, respectively; digits are rounded.}}
\label{tab:gamma}
{\color{revisioncyan}\footnotesize
\begin{tabular}{c@{\quad}c@{\quad}c@{\qquad}
                c@{\quad}c@{\quad}c}
\toprule
$k$ & $\eta_k$ & $\gamma_k$ & $k$ & $\eta_k$ & $\gamma_k$\\
\midrule
2 & $0.250000000000$ & $0.500000000000$ & 6 & $0.081523681041$ & $0.081955364899$\\
3 & $0.165164800491$ & $0.166666666667$ & 7 & $0.069728502945$ & $0.069728504012$\\
4 & $0.123133086084$ & $0.124103599697$ & 8 & $0.060912117517$ & $0.060914841075$\\
5 & $0.098109482185$ & $0.098111200886$ & 9 & $0.054073411933$ & $0.054073411934$\\
\bottomrule
\end{tabular}}
\end{table}

\begin{corollary}[Asymptotics]
\label{cor:gamma-asymptotic}
Let $\varphi=(1+\sqrt5)/2$.  Then
\begin{equation}
\label{eq:gamma-asymptotic}
 \lim_{k\to\infty}k\gamma_k
 =\operatorname{arsinh}\!\left(\frac12\right)
 =\log\varphi.
\end{equation}
\end{corollary}

\begin{proof}
Put $a=c/k$ and define the odd binomial sum
\[
 S_k(c):=\sum_{\substack{0\le j\le k\\ j\ \mathrm{odd}}}
 \binom{k}{j}\left(\frac ck\right)^j
 =\frac12\left[\left(1+\frac ck\right)^k
                 -\left(1-\frac ck\right)^k\right].
\]
Hence $S_k(c)\to\sinh c$ uniformly for $c$ in every compact interval.  If
$k$ is odd, \cref{eq:pk-odd} gives exactly
\[
 p_k(c/k)=1-2S_k(c)+2(c/k)^k.
\]
If $k=2m$ is even, pairing the odd terms in $S_k$ around the middle layer
shows that
\begin{align*}
 p_k(c/k)={}&1-2S_k(c)+2(c/k)^k\\
 &+\sum_{\substack{1\le j<m\\j\ \mathrm{odd}}}
   \binom{k}{j}(c/k)^{k-j}
 +\mathbf1_{\{m\ \mathrm{odd}\}}
   \frac12\binom{k}{m}(c/k)^m.
\end{align*}
For $0\le c\le M$, all terms after $1-2S_k(c)$ are bounded by a
constant multiple of $2^k(M/k)^{k/2}$, and therefore converge uniformly to
zero.  Consequently,
\[
 p_k(c/k)\longrightarrow1-2\sinh c
\]
uniformly on compact $c$-intervals.

Let $c_0=\operatorname{arsinh}(1/2)$.  Fix
$0<\varepsilon<c_0$.  Uniform convergence gives, for all sufficiently
large $k$,
\[
 p_k\!\left(\frac{c_0-\varepsilon}{k}\right)>0,
 \qquad
 p_k\!\left(\frac{c_0+\varepsilon}{k}\right)<0.
\]
Since $p_k$ is strictly decreasing, its zero satisfies
$c_0-\varepsilon<k\gamma_k<c_0+\varepsilon$.  Thus
$k\gamma_k\to c_0=\log\varphi$.  The same calculation applied to
$r_k(c/k)=1-2S_k(c)$ shows that the earlier constants also satisfy
$k\eta_k\to c_0$; the gain in \cref{eq:pk-prior-comparison} is a
finite-$k$ sharpening.
\end{proof}

\begin{corollary}[Finite-copy NPT intervals]
\label{cor:finite-copy-NPT}
If $d>1/\gamma_k$, then every Werner state with
\begin{equation}
\label{eq:finite-copy-NPT}
 -\gamma_k\le\alpha<-\frac1d
\end{equation}
is simultaneously NPT and $k$-copy undistillable; in particular, this
interval is nonempty.
\end{corollary}

The constants $\gamma_k$ are certified lower bounds for the unknown sharp
$k$-copy threshold, not claims of optimality when $k\ge3$.  For comparison,
Pablo Costa Rico proved endpoint positivity for a nontrivial three-copy subclass:
self-adjoint rank-two operators with one positive and one negative eigenvalue
\cite[Proposition~3]{CostaRico2025}.  The general rank-two three-copy form,
and hence the entire interval between $-1/2$ and $-\gamma_3=-1/6$, remains
unresolved by the present argument.

\section{Kronecker-sum consequences}
\label{sec:kronecker}

The traceless partial-trace estimate has a direct dual formulation.  Let
$m=\dim U$, $n=\dim V$, and let $A\in\cL(U)$ and $B\in\cL(V)$ be traceless.
Set
\begin{equation}
\label{eq:kronecker-M}
 M=A\otimes I_V+I_U\otimes B.
\end{equation}
For every $Z\in\cL(U\otimes V)$,
\begin{equation}
\label{eq:ptr-duality}
 \ip{Z}{M}_{\HS}
 =\ip{\Pi_U^0(\Tr_VZ)}{A}_{\HS}
 +\ip{\Pi_V^0(\Tr_UZ)}{B}_{\HS}.
\end{equation}

\begin{corollary}[Sharp dimension-dependent Kronecker-sum bound]
\label{cor:general-kronecker}
Assume $m,n\ge2$.  For traceless $A\in M_m(\C)$ and $B\in M_n(\C)$,
\begin{equation}
\label{eq:general-kronecker}
 s_1(M)^2+s_2(M)^2
 \le\kappa_{m,n}\bigl(\norm{A}_2^2+\norm{B}_2^2\bigr),
\end{equation}
where
\begin{equation}
\label{eq:kappa}
 \kappa_{m,n}
 =\begin{cases}
 2, & \displaystyle \frac1m+\frac1n\ge\frac12,\\[6pt]
 \displaystyle 3-\frac2m-\frac2n,
 & \displaystyle \frac1m+\frac1n<\frac12.
 \end{cases}
\end{equation}
For every pair $(m,n)$, the constant $\kappa_{m,n}$ is attained.
\end{corollary}

\begin{proof}
Let $Z$ have rank at most two and $\norm{Z}_2=1$.  By
\cref{eq:ptr-duality}, Cauchy--Schwarz in a direct sum, and
\cref{eq:traceless-ptr},
\begin{align}
 \abs{\ip{Z}{M}_{\HS}}^2
 &\le\bigl(\norm{A}_2^2+\norm{B}_2^2\bigr)
 \Bigl[
 \norm{\Pi_U^0(\Tr_VZ)}_2^2
 +\norm{\Pi_V^0(\Tr_UZ)}_2^2
 \Bigr]\notag\\
 &\le\bigl(\norm{A}_2^2+\norm{B}_2^2\bigr)
 \left[2+
 \left(\frac12-\frac1m-\frac1n\right)\abs{\Tr Z}^2
 \right].
 \label{eq:kronecker-intermediate}
\end{align}
If the coefficient of $\abs{\Tr Z}^2$ is nonpositive, the bracket is at most
$2$.  If it is positive, \cref{eq:trace-rank} gives
$\abs{\Tr Z}^2\le2$, so the bracket is at most
$3-2/m-2/n$.  Taking the maximum over rank-two unit $Z$ and using
\cref{lem:kyfan-duality} proves the inequality.

It remains to show sharpness.  Suppose first that
$1/m+1/n\ge1/2$.  Choose orthonormal $v,w\in\C^n$, take
$A=0$ and $B=\ketbra{v}{w}$.  Then $B$ is traceless,
$\norm{B}_2=1$, and the singular value $1$ of $B$ is repeated $m$ times in
$I_m\otimes B$.  Hence the left side of \cref{eq:general-kronecker} is $2$.

Now suppose that $1/m+1/n<1/2$.  Let $P$ be a rank-two orthogonal projection
in $\C^m$ and let $Q=\ketbra{v}{v}$ be a rank-one orthogonal projection in
$\C^n$.  Set
\begin{equation}
\label{eq:general-sharp-AB}
 A=P-\frac2mI_m,
 \qquad
 B=2\left(Q-\frac1nI_n\right).
\end{equation}
Both matrices are traceless, and, with $\kappa=\kappa_{m,n}$,
\[
 \norm{A}_2^2+\norm{B}_2^2
 =6-\frac4m-\frac4n=2\kappa.
\]
The Kronecker sum is Hermitian.  On $\ran(P)\otimes\ran(Q)$ it has the
eigenvalue $\kappa$ with multiplicity two.  Its remaining eigenvalues are
$\kappa-1$, $\kappa-2$, and $\kappa-3$ (with the evident
multiplicities), all of smaller absolute value because $\kappa>2$.
Therefore its two largest singular values are both $\kappa$, and both sides
of \cref{eq:general-kronecker} equal $2\kappa^2$.
\end{proof}

Among equal local dimensions, the two-ququart case is exactly the point at
which the scalar correction in \cref{eq:traceless-ptr} cancels.

\begin{corollary}[Sharp $4\times4$ Kronecker-sum inequality]
\label{cor:kronecker-four}
For traceless $A,B\in M_4(\C)$,
\begin{equation}
\label{eq:kronecker-four}
 s_1(A\otimes I_4+I_4\otimes B)^2
 +s_2(A\otimes I_4+I_4\otimes B)^2
 \le2\bigl(\norm{A}_2^2+\norm{B}_2^2\bigr).
\end{equation}
The constant is attained.
\end{corollary}

\begin{proof}
The inequality is the case $m=n=4$ of
\cref{cor:general-kronecker}.  For equality, take
\begin{equation}
\label{eq:sharp-AB}
 A=B=\frac14\operatorname{diag}(1,-1,0,0).
\end{equation}
Then $\norm{A}_2^2+\norm{B}_2^2=1/4$, while the two largest singular values of
$A\otimes I_4+I_4\otimes B$ are both $1/2$.  Both sides of
\cref{eq:kronecker-four} equal $1/2$.
\end{proof}

\section{Discussion and outlook}
\label{sec:discussion}

The two-copy proof isolates a dimension-free geometric chain:
\begin{align*}
 &\text{rank-two symmetric projection bound}
 \Longrightarrow \text{double-skew Ky Fan bound}\\
 &\hspace{2cm}\Longrightarrow \text{maximally mixed qubit positivity}
 \Longrightarrow \text{rank-two partial-trace inequality}.
\end{align*}
The endpoint $\alpha=-1/2$ turns the last inequality into the two-copy Werner
form, giving the exact threshold.  The same partial-trace language remains
useful for arbitrary $k$, but the geometry changes qualitatively: the
operator $H_k(\psi)$ has a full alternating hierarchy of marginals, rather
than the single double-antisymmetric sector controlled at $k=2$.

For fixed $d\ge2$ and $k\ge1$, the exact formulations in
\cref{sec:k-formulations} identify three equivalent targets for future work:
\begin{equation}
\label{eq:three-targets}
\begin{aligned}
 &q_k\!\left(-\frac12,C\right)\ge0
 &&\text{for every }C\in\cL((\C^d)^{\otimes k})
   \text{ with }\rank C\le2\\
 &\qquad\Longleftrightarrow\
 \Phi_d^{\otimes k}\text{ is $2$-positive}
 \ \Longleftrightarrow\
 H_k(\psi)\succeq0
 &&\text{for every admissible }\psi,
\end{aligned}
\end{equation}
where in the last condition $U_1=\cdots=U_k=\C^d$.
They also prevent a misleading induction.  The map
$\Phi_d^{\otimes2}$ is $2$-positive, but tensoring it with another copy of
$\Phi_d$ is precisely the new assertion to be proved; it cannot be assumed.
Likewise, the recursion for $q_k$ applies the lower-copy form to partial
traces whose rank may be arbitrarily larger than two.

Two positive results survive tensorization.  The exact factorization theorem
shows that any negative witness cannot be reduced, through a tensor
decomposition, to a rank-two factor on only one or two copy blocks.  The
complement-pair bounds sharpen the dimension-free finite-copy intervals of
\cite{QiGupurWuGuo2024}; for sufficiently large $d$ these intervals contain
NPT states.  The asymptotic scale
$\gamma_k\sim(\log\varphi)/k$, which is shared by the earlier bounds,
quantifies the output of this complement-pair argument.

The central unresolved question is still the endpoint
$q_k(-1/2,C)\ge0$ for unrestricted rank-two $C$ and $k\ge3$.  A proof for all
$k$ would establish NPT bound entanglement in the Werner family.  A
counterexample for some finite $k$ would instead exhibit strictly collective
distillability requiring more than two copies.  The three-copy operator
\cref{eq:H3}, together
with the copy-irreducibility constraint \cref{cor:copy-irreducible}, is a
concrete next target.

Other directions remain worthwhile independently of the many-copy problem.
A classification of equality in the double-skew and partial-trace inequalities
would clarify the exterior-algebra mechanism behind saturated
Kronecker-sum cases \cite{LiuChen2026}.  A direct sum-of-squares factorization
of $H_2(\psi)$ could also suggest which alternating marginal combinations in
$H_3$ admit geometric control.

\section*{Acknowledgments}
We thank Felix Huber for promptly sharing the concurrent manuscript.

\section*{AI-use statement}
OpenAI's ChatGPT 5.6 Sol provided an initial proof.  The authors reconstructed and revised the
arguments and take responsibility for the mathematical content.  Further
details are recorded in the provenance note.

\end{document}